\documentclass[conference]{IEEEtran}
\IEEEoverridecommandlockouts
\usepackage{cite}
\usepackage{amsmath,amssymb,amsfonts, amsthm}
\usepackage{mathtools}
\usepackage{algorithmic}
\usepackage{graphicx}
\usepackage{textcomp}
\usepackage{xcolor}
\def\BibTeX{{\rm B\kern-.05em{\sc i\kern-.025em b}\kern-.08em
    T\kern-.1667em\lower.7ex\hbox{E}\kern-.125emX}}

\usepackage{mathrsfs}
\usepackage{url}

\newtheorem{definition}{Definition}[]
\newtheorem{proposition}{Proposition}[]
\newtheorem{theorem}{Theorem}[]
\newtheorem{corollary}{Corollary}
\newtheorem{remark}{Remark}[]
\newtheorem{lemma}{Lemma}[]

\newcommand{\Z}{\mathbb{Z}}
\newcommand{\Q}{\mathbb{Q}}
\newcommand{\C}{\mathbb{C}}

\newcommand{\R}{\mathbb{R}}

\newcommand{\E}{\mathbb{E}}

\newcommand{\vor}{\mathcal{V}}

\newcommand{\Trd}{\text{Trd}}
\newcommand{\Nrd}{\text{Nrd}}

\begin{document}

%\title{Finite-Blocklength Analysis of Quaternionic Alamouti Codes over Eisenstein Integers\\
\title{Finite-Blocklength Analysis of Alamouti Codes over Eisenstein Integers \\
%{\footnotesize \textsuperscript{*}Note: Sub-titles are not captured for https://ieeexplore.ieee.org  and
\thanks{The work of J.G.F. Souza is supported by the São Paulo Research Foundation (FAPESP), Brazil, under grant 2025/22940-2. The work of Y.-C. Huang was supported by the National Science and Technology Council, Taiwan, under Grant NSTC 113-2223-E-A49-005 -MY3}
}

\author{\IEEEauthorblockN{Juliana G. F. Souza}
\IEEEauthorblockA{\textit{Institute of Mathematics, Statistics}\\ 
\textit{and Scientific Computing, UNICAMP}\\
%São Paulo, Brazil \\
julianagfs@ime.unicamp.br}
\and
\IEEEauthorblockN{Yu-Chih Huang}
\IEEEauthorblockA{\textit{Institute of Communications Engineering,}\\
\textit{National Yang Ming Chiao Tung University}\\
%Hsinchu, Taiwan \\
jerryhuang@nycu.edu.tw}
% \and
% \and
% \and
% \and
}

\maketitle

\begin{abstract}
We study a space--time block code from a maximal order in the definite quaternion algebra $(-1,-3)_{\Q}$. Its embedding into $\C^{2\times 2}$ yields an Alamouti--Eisenstein code over $\Z[w]$ with full diversity, orthogonality, and non-vanishing determinant. The underlying lattice is isomorphic to $\Z[w]^2$, while the embedded lattice has $A_2\oplus A_2$ geometry, yielding a hexagonal shaping gain. We compare it with the classical Alamouti code over $\Z[i]$ in terms of shaping, constellation-constrained mutual information, and finite-blocklength achievable rates, obtaining an asymptotic energy gain of about $0.79$~dB and a small but positive mutual-information gain. At the same SNR and rate, the Alamouti--Eisenstein design also improves short-packet reliability.
\end{abstract}


\section{Introduction}

Space--time block codes (STBCs) are widely used in MIMO systems to obtain diversity with moderate decoding complexity. The $2\times 2$ Alamouti code is the classical example. Its orthogonal structure gives full diversity with simple symbol-wise ML decoding, and it is widely used in standards and prototypes~\cite{alamouti1998,tarokh1999stbc,paulraj2003}. In practice, it is often paired with square QAM carved from the Gaussian integers $\Z[i]$.

The Alamouti code also admits a natural quaternionic description: each $2\times 2$ block comes from a Hamilton quaternion under a suitable embedding, and the determinant equals the reduced norm~\cite{oggier2007cyclic}. This connects STBC design with algebraic constructions from natural and maximal orders in definite quaternion algebras, and with lattice constellations having better packing efficiency than $\Z^4$~\cite{oggier2006perfect,elia2006explicit,hollanti2008maximal,vehkalahti2009densest,agrell2009,karlsson2011,welti1974,zetterberg1977,alves2015lattices,jsouza2024multilevel}.

There is growing interest in non-square constellations and finite-blocklength performance in the context of 6G~\cite{3gpp:RP253877,itur_m2160,itur_m2516}. More broadly, the 6G literature emphasises robust MIMO with low processing delay and practical complexity (e.g.,~\cite{andrews20246g}). In this setting, orthogonal space--time schemes such as the Alamouti code remain a natural baseline.

This also makes short-packet analysis relevant. In this regime, the condition $R<C$ is not sufficient, and achievable rates depend on mutual information and channel dispersion through finite-blocklength bounds~\cite{polyanskiy2010finite,durisi2016short-packet}. Alamouti-type transmit diversity is a common baseline in short-packet settings~\cite{durisi2016mURLL,yuan2021noma}. At the same time, 6G reports highlight constellation shaping and non-rectangular constellations as practical tools to improve efficiency~\cite{3gpp:RP253877,itur_m2516}. This motivates hexagonal constellations over the Eisenstein integers $\Z[w]$, which offer a shaping advantage over square QAM while preserving a clean algebraic structure for lattice-based space--time designs~\cite{freudenberger2017eisenstein,oggier2007cyclic}.

HEX constellations and broader algebraic or embedded Alamouti-type constructions have appeared previously in the space--time coding literature~\cite{oggier2007cyclic,hollanti2008maximal,sinnokrot2008alamouti}. In particular, \cite{sinnokrot2008alamouti} develops a broader embedded Alamouti framework for high-rate space--time codes with reduced decoding complexity. Here, rather than proposing a new STBC family, we study the classical $2\times2$ Alamouti orthogonal design with symbols in $\Z[w]$, realised via an explicit maximal order in the definite quaternion algebra $B=(-1,-3)_{\Q}$. This gives a concrete $2\times2$ specialisation within broader algebraic STBC frameworks, while making explicit the reduced-norm structure, the underlying $A_2\oplus A_2$ lattice geometry, and the resulting shaping, mutual-information, and finite-blocklength comparison with the standard Gaussian Alamouti baseline. Our aim is to isolate the effect of replacing the QAM scalar alphabet by the Eisenstein/HEX one within the orthogonal Alamouti design.

Within this framework, we describe the \emph{Alamouti--Eisenstein} realisation from a maximal order $\Gamma \subset B$ with $\Gamma \cong \Z[w]\oplus i\,\Z[w]$ as a $\Z$-module under the canonical embedding $B \hookrightarrow \C^{2\times 2}$. The resulting matrices have entries in $\Z[w]$, determinant equal to the reduced norm, and preserve the orthogonal Alamouti structure, so the usual symbol-wise ML decoding still applies. This also makes explicit the connection between the shaping gain and the geometry, and enables comparison with the classical Alamouti code over $\Z[i]$ in terms of shaping, constellation-constrained mutual information, and finite-blocklength normal-approximation bounds.

The paper is organised as follows. Section~\ref{sec:preliminaries} reviews quaternion algebras and lattices for $B = (-1,-3)_{\Q}$ and its maximal order $\Gamma$. Section~\ref{sec:ST-codes} revisits the Alamouti scheme from a quaternionic viewpoint and introduces the Alamouti--Eisenstein code. Section~\ref{sec:discussions} compares shaping, mutual information, and finite-blocklength performance with the classical Alamouti code. Section~\ref{sec:conclusions} concludes and outlines future work.

\section{Preliminaries}\label{sec:preliminaries}

In this section we recall basic facts on quaternion algebras and maximal orders, with emphasis on the definite algebra $B = (-1,-3)_{\Q}$, then review standard notions on Euclidean lattices.

\subsection{Quaternion Algebra}

Let $K$ be a field of characteristic $\neq 2$. For $a,b\in K^{\times}$ the quaternion $K$-algebra
\[
B=(a,b)_K = K\langle i,j \mid i^2=a,\ j^2=b,\ ij=-ji\rangle
\]
is a central simple $K$-algebra of dimension $4$ with basis $\{1,i,j,ij\}$\cite{voight2021quaternion}. For $q=x+yi+zj+tij\in B$ we define the conjugate, reduced norm and reduced trace by \( q^* = x-(yi+zj+tij),\ \ \Nrd_B(q)=q q^* = x^2-a y^2-b z^2+ab t^2,\ \ \Trd_B(q)=q+q^* = 2x, \) respectively. Conjugation is $K$-linear, reverses products, and $\Nrd_B$ is multiplicative.

Viewing $B$ as a two-dimensional vector space over $L=K(j)\cong K(\sqrt{b})$ with basis $(1,i)$, left multiplication by $q$ gives a $2\times 2$ matrix over $L$.

\begin{lemma}[\cite{voight2021quaternion, berhuy2013central}]\label{matrixform}
Let $B=(a,b)_K$ and $L=K(j)$. For $q=x+yi+zj+tij\in B$, the matrix of left multiplication by $q$ in the $L$-basis $(1,i)$ is
\[
\mathbf{M}_q =
\begin{pmatrix}
x + z\sqrt{b} & a(y-t\sqrt{b})\\
y + t\sqrt{b} & x - z\sqrt{b}
\end{pmatrix}\in M_2(L),
\]
and the map $\varphi_{B,L}:q\mapsto \mathbf{M}_q$ is an injective $K$-algebra morphism satisfying
$\det(\mathbf{M}_q)=\Nrd_B(q)$ for all $q\in B$.
\end{lemma}

The Hamilton quaternions arise for $K=\R$ and $a=b=-1$: \(
\mathbb{H}=(-1,-1)_\R
 = \{x+yi+zj+tij:\ x,y,z,t\in\R\}.
\)
Writing $\C=\R\oplus \R i$, every $q\in\mathbb{H}$ decomposes uniquely as $q=x_0+x_1j$ with $x_0,x_1\in\C$.

\begin{corollary}[\cite{berhuy2013central}]\label{hamilon_embedding}
The embedding $\varphi_{\mathbb{H},\C}:\mathbb{H}\hookrightarrow M_2(\C)$ induced by Lemma~\ref{matrixform} is
\[
q=x_0+jx_1 \longmapsto
\begin{pmatrix}
x_0 & -x_1^*\\
x_1 & x_0^*
\end{pmatrix},\qquad x_0,x_1\in\C,
\]
with $\det(\varphi_{\mathbb{H},\C}(q))=\Nrd_{\mathbb{H}}(q)$. Thus $\mathbb{H}$ is isomorphic (as an $\R$-algebra) to the real subalgebra of $M_2(\C)$ consisting of such matrices.
\end{corollary}

%\subsection{Maximal Order in $B=(-1,-3)_\Q$}

We now restrict to quaternion algebras over $\Q$. An \emph{order} $\mathcal{O}$ in $B=(a,b)_\Q$ is a subring of $B$ containing $1$ that is a free $\Z$-module of rank~$4$ (equivalently, a full $\Z$-lattice in $B$). An order is \emph{maximal} if it is not properly contained in any larger order.

The definite quaternion algebra of interest is \( B=(-1,-3)_\Q. \)
A maximal order in $B$ is given by (see, e.g.,~\cite{voight2021quaternion})
\[
\Gamma = \Z[w]\ \oplus\ i\,\Z[w],\qquad
w=\tfrac{1+\sqrt{-3}}{2},
\]
where $\Z[w]$ is the ring of Eisenstein integers. Every $q\in\Gamma$ can thus be written uniquely as \(q=x_0+i x_1,\ \ x_0,x_1\in\Z[w]. \) Conjugation restricts to $\Gamma$, and the reduced norm splits as
\(\Nrd(q)=q q^* = x_0 x_0^* + x_1 x_1^* = N(x_0)+N(x_1), \)
where $N$ is the usual norm on $\Z[w]$. Writing $x_0=x+zw$ and $x_1=y+tw$ with $x,y,z,t\in\Z$, we obtain
\(
\Nrd(q) = (x^2+xz+z^2) + (y^2+yt+t^2),
\)
so the $\Z$-lattice $\Gamma$ with quadratic form $\Nrd$ is isometric to $A_2\oplus A_2$. %The unit group $\Gamma^\times$ has $12$ elements; we only use that $|\Gamma^\times|<\infty$.

\subsection{Lattices}
 
A lattice $\Lambda$ in $\R^n$ is an additive discrete subgroup of $\R^n$, i.e.,
\[
  \Lambda = \left\{ \sum_{i=1}^n z_i \mathbf{v}_i : z_i \in \Z \right\},
\]
for linearly independent vectors $\mathbf{v}_1,\ldots,\mathbf{v}_n\in\R^n$.
The matrix $\mathbf{B}$ whose columns are the $\mathbf{v}_i$ is called a
generator matrix of $\Lambda$. The Voronoi region of $\Lambda$ at $\mathbf{x}\in\Lambda$ is
\[
  \vor_\Lambda(\mathbf{x})
  = \big\{ \mathbf{y} \in \R^n : \|\mathbf{y} - \mathbf{x}\|
     \leq \|\mathbf{y} - \lambda\|,\ \forall\, \lambda \in \Lambda \big\}.
\]
In particular, $\vor_\Lambda(0)$ tiles $\R^n$ under translations by $\Lambda$
and serves as the basic shaping cell.

If $\Lambda_c\subseteq \Lambda_f$ are full-rank lattices, the pair is \emph{nested}, with $\Lambda_f$ the \emph{fine} and $\Lambda_c$ the \emph{coarse} lattice. The associated Voronoi constellation (lattice code) is \(\Lambda_f/\Lambda_c := \Lambda_f \cap \vor_{\Lambda_c}(0), \) where $\Lambda_f$ is the coding lattice and $\Lambda_c$ provides shaping.  Over an AWGN channel, good packing of $\Lambda_f$ gives better-separated points in $\Lambda_f/\Lambda_c$ and reduces the error probability.

\section{Space--Time Codes}\label{sec:ST-codes}

We briefly recall the MIMO channel model and the standard rank/determinant criteria \cite{tarokh1998criteria,tse2005fundamentals}. We then review the classical Alamouti code from a quaternionic viewpoint \cite{oggier2004algebraic,oggier2007cyclic} and introduce its Eisenstein (HEX) counterpart derived from the maximal order of $B=(-1,-3)_{\Q}$.

We consider a quasi-static, flat Rayleigh fading channel with $n_t$ transmit and
$n_r$ receive antennas. Over one fading block of $T$ channel uses,
\[
\mathbf{Y}=\mathbf{H}\mathbf{X}+\mathbf{Z},
\]
where $\mathbf{X}\in\C^{n_t\times T}$ is the transmitted block, $\mathbf{Y}$ is
the received block, $\mathbf{H}$ has i.i.d.\ entries $h_{ij}\sim\mathcal{N}_c(0,1)$
and $\mathbf{Z}$ has i.i.d.\ entries $\sim\mathcal{N}_c(0,N_0)$. We take $T$ equal
to the coherence time, so $\mathbf{H}$ is constant within the block. A
space--time block code (STBC) is a finite set $\mathcal{C}$ of such matrices. In the coherent case, ML decoding chooses \(
\widehat{\mathbf{X}}=\arg\min_{\mathbf{X}\in\mathcal{C}}
\|\mathbf{Y}-\mathbf{H}\mathbf{X}\|_F^2.
\)

The codeword error probability $P_e$ can be bounded by a union bound over
pairwise error events \(
P_e\le \frac{1}{|\mathcal{C}|}\sum_{\mathbf{X}\in\mathcal{C}}
\sum_{\mathbf{\widehat{X}}\neq \mathbf{X}} P(\mathbf{X}\to \mathbf{\widehat{X}}),
\)
where $P(\mathbf{X}\to \mathbf{\widehat{X}})$ is the ML pairwise error
probability (PEP). Under i.i.d.\ Rayleigh fading, one obtains the determinant
bound~\cite[Thm.~1]{tarokh1998criteria}
\begin{equation}\label{error_rayleigh}
P(\mathbf{X}\to \mathbf{\widehat{X}})
\le \det\!\left[\mathbf{I}_{n_t}
+\frac{(\mathbf{X}-\mathbf{\widehat{X}})(\mathbf{X}-\mathbf{\widehat{X}})^{\dagger}}{4N_0}\right]^{-n_r}.
\end{equation}

Let $r=\operatorname{rank}(\mathbf{X}-\mathbf{\widehat{X}})$. For high SNR the
PEP decays roughly as $\mathrm{SNR}^{-r n_r}$, so the diversity order is
$r n_r$. The best possible diversity is $n_t n_r$, obtained when
$r=n_t$ for all distinct pairs; in this case the code has \emph{full diversity}.
The dominant term in the union bound is then governed by the minimum determinant
\[
\delta_{\min}
=\min_{\mathbf{X}\neq \mathbf{\widehat{X}}}
\det\!\big((\mathbf{X}-\mathbf{\widehat{X}})
(\mathbf{X}-\mathbf{\widehat{X}})^{\dagger}\big),
\]
and $(\delta_{\min})^{1/n_t}$ is often called the coding gain.

In this work we restrict attention to \emph{linear} STBCs, for which
$\mathbf{X},\mathbf{X}'\in\mathcal{C} \Rightarrow
\mathbf{X}\pm\mathbf{X}'\in\mathcal{C}$. Then it is enough to consider
differences from the zero matrix, \( \delta_{\min}=\min_{\mathbf{X}\neq 0}\det(\mathbf{X}\mathbf{X}^{\dagger}), \)
and the induced lattice structure allows standard sphere-decoding techniques for
ML detection.

\subsection{Classical Alamouti Code}

The $2\times 2$ Alamouti code~\cite{alamouti1998} encodes two complex symbols $x_0,x_1\in\C$ as
\[
\mathcal{C}
=
\left\{
\begin{pmatrix}
x_0 & -x_1^*\\
x_1 & x_0^*
\end{pmatrix}
\;\middle|\;
x_0,x_1\in\C
\right\}.
\]
For any distinct $\mathbf{X},\mathbf{X}'\in\mathcal{C}$,
\(
\det(\mathbf{X}-\mathbf{X}') = |x_0-x_0'|^2+|x_1-x_1'|^2>0,
\)
so the code has full diversity, and
\(
\mathbf{X}^\dagger\mathbf{X}=(|x_0|^2+|x_1|^2)\mathbf{I}_2,
\)
which yields an orthogonal design and symbol-wise ML decoding.

These properties admit a quaternionic interpretation. Writing Hamilton’s algebra as
\(
\mathbb{H}=(-1,-1)_{\R}=\{x_0+jx_1 \,|\, x_0,x_1\in\C\},
\)
each codeword $\mathbf{X}\in\mathcal{C}$ is the image of $q=x_0+jx_1$ under the embedding
$\varphi_{\mathbb{H},\C}$ of Corollary~\ref{hamilon_embedding}, and
\(
\det(\varphi_{\mathbb{H},\C}(q))=\Nrd(q),
\)
the reduced norm of $q$. Since $\mathbb{H}$ is a division algebra, $\Nrd(q)\neq 0$ for $q\neq 0$, which recovers full diversity~\cite{oggier2004algebraic,oggier2007cyclic}.

For algebraic constructions we restrict scalars to $\Q$ and use the rational Hamilton algebra.
Let $\mathcal{O}=\Z[i]\oplus j\Z[i]\subset\mathbb{H}_{\Q}$ be a natural order~\cite{reiner1975maximal}.
The (classical) Alamouti code is the image of $\mathcal{O}$ under $\varphi_{\mathbb{H}_\Q,\Q(i)}$, and
finite constellations are obtained by choosing a shaping set $A\subset\Z[i]$ and taking
\[
\mathcal{C}_{\Z[i]}
=
\left\{
\begin{pmatrix}
x_0 & -x_1^*\\
x_1 & x_0^*
\end{pmatrix}
\;\middle|\;
x_0,x_1\in A
\right\}.
\]
In this rational model the determinant still equals the reduced norm, so full diversity and orthogonality are preserved.

\subsection{Alamouti--Eisenstein Code}

Let $B=(-1,-3)_{\Q}$ and let $\Gamma$ be the maximal order
\(
\Gamma = \Z[w] \ \oplus\ i\,\Z[w],\ \ 
w = \frac{1 + \sqrt{-3}}{2},
\)
where $\Z[w]$ is the ring of Eisenstein integers. Every $q\in\Gamma$ has a unique decomposition
$q = x_0 + i x_1$ with $x_0,x_1\in\Z[w]$ and $ix = x^* i$ for $x\in\Q(\sqrt{-3})$. This is the Eisenstein analogue of the Hamiltonian order $\Z[i]\oplus j\Z[i]$ and naturally leads to an Alamouti-type design. Applying Lemma~\ref{matrixform} with $a=-1$, $b=-3$ and $L=\Q(\sqrt{-3})$ gives the $L$-algebra embedding
\[
\varphi_{B,L}: B \longrightarrow \mathrm{M}_2(L),\qquad
q=x_0+i x_1 \longmapsto
\begin{pmatrix}
x_0 & -x_1^* \\
x_1 & x_0^*
\end{pmatrix}.
\]
Composing with a complex embedding $\sigma:L\hookrightarrow\C$ yields explicit code matrices
\(
\mathbf{X} \!=\!
\begin{pmatrix}
\sigma(x_0) & -\sigma(x_1)^* \\
\sigma(x_1) & \sigma(x_0)^*
\end{pmatrix}, x_0,x_1\in\Z[w].
\)

As in the classical Alamouti code, one checks that
\[
\mathbf{X}^{\dagger}\mathbf{X}
= \big(|\sigma(x_0)|^2 + |\sigma(x_1)|^2\big)\mathbf{I}_2,
\qquad
\det(\mathbf{X}) = \Nrd(q),
\]
so the columns are orthogonal with equal energy and the determinant equals the reduced norm. Since $B$ is a division algebra, $\Nrd(q)\neq 0$ for $q\neq 0$, and the resulting STBC has full diversity and the non-vanishing determinant (NVD) property~\cite{oggier2007cyclic, oggier2004algebraic}, while still allowing simple symbol-wise ML decoding.The decoding rule remains symbol-wise as in the classical Alamouti code, with the main practical difference coming from the use of HEX rather than rectangular QAM signalling.

Choosing a finite shaping set $\mathcal{X}\subset\Z[w]$ (e.g., a hexagonal box of bounded energy) and normalising the average energy gives the finite codebook
\begin{equation}\label{eq:AE-matrix}
\mathcal{C}_{\Z[w]}(\mathcal{X}) = \left\{
\begin{pmatrix}
x_0 & -x_1^* \\
x_1 & x_0^*
\end{pmatrix}
\;\middle|\;
x_0,x_1\in \mathcal{X}
\right\}.
\end{equation}
Compared with the Gaussian Alamouti code, this \emph{Alamouti--Eisenstein code} replaces the underlying square lattice by the denser Eisenstein lattice, thereby improving packing/shaping properties while preserving full diversity, NVD, and low-complexity decoding.

% \begin{remark}
% The matrix form above is the classical Alamouti design with symbols in $\Z[w]$.  We use $B=(-1,-3)_{\Q}$ to obtain it canonically from the maximal order $\Gamma=\Z[w]\oplus i\,\Z[w]$, which fixes the embedded lattice ($A_2\oplus A_2$ geometry) and yields $\det(\mathbf{X})=\Nrd(q)\in\Z$ for $q\in\Gamma$.
% \end{remark}

\begin{remark}
The matrix form above is the classical Alamouti design with symbols in $\Z[w]$. The quaternion-algebra viewpoint provides a structural derivation via the maximal order $\Gamma=\Z[w]\oplus i\,\Z[w]$: it fixes the embedded lattice ($A_2\oplus A_2$ geometry) and gives $\det(\mathbf{X})=\Nrd(q)\in\Z$ for $q\in\Gamma$.
\end{remark}

\section{Performance Analysis}\label{sec:discussions}

This section compares the proposed Alamouti--Eisenstein scheme with the classical Alamouti code over $\Z[i]$.  We study shaping via codeword error rate (CER), constellation-constrained mutual information, and finite-blocklength normal-approximation bounds for the equivalent $2\times 1$ Alamouti channel. Throughout, $\mathcal C_{\Z[i]}$ and $\mathcal C_{\Z[w]}$ denote scalar-shaped constellations in $\Z[i]$ and $\Z[w]$, respectively; for a given scalar alphabet $\mathcal X$, the associated Alamouti codebooks have cardinality $|\mathcal X|^2$.

\subsection{Shaping Gain}

We compare the classical Alamouti design (over $\Z[i]$) with the Alamouti–Eisenstein design (over $\Z[w]$). In both cases the transmitted block has the Alamouti form
\[
\mathbf{X}=\begin{pmatrix}
x_0 & -x_1^*\\
x_1 & x_0^*
\end{pmatrix}, \qquad x_0,x_1\in\mathcal{R},
\]
where $\mathcal{R}$ is either a finite subset of $\Z[i]$ (square shaping) or of $\Z[w]$ (hexagonal shaping)~\cite{alamouti1998,tse2005fundamentals,tarokh1999stbc}. For any $\mathbf{X},\widehat{\mathbf X}$ in the code, the difference $\Delta \mathbf{X}=\mathbf{X}-\widehat{\mathbf X}$ has the same structure and a direct calculation gives \( \Delta\mathbf{X}\Delta \mathbf{X}^{\dagger} =(|\Delta x_0|^2+|\Delta x_1|^2)\mathbf{I}_2, \) so \( \det\!\big((\Delta \mathbf{X})(\Delta \mathbf{X})^{\dagger}\big) =(|\Delta x_0|^2+|\Delta x_1|^2)^2, \) that is, orthogonality “scalarises” the determinant (see \cite[Sec.~3.3]{tse2005fundamentals}).

\begin{proposition}[Shaping gain]\label{prop:shaping_gain}
Consider two Alamouti constellations with Voronoi shaping, one based on $\Z[i]\cong\Z^2$ (square) and the other on $\Z[w]\cong A_2$ (hexagonal), with same minimum squared block distance $\Delta_{\min}$ and the same union–bound constant $C$. Then, for any fixed target codeword error rate $\varepsilon\in(0,C)$ in the high-SNR regime, the hexagonal constellation enjoys an asymptotic shaping gain of
\[
\Delta\mathrm{SNR}
=10\log_{10}\frac{6}{5}\approx 0.79~\mathrm{dB}.
\]
\end{proposition}

\begin{proof}
From the union--bound approximation for the Alamouti code over Rayleigh fading
with one receive antenna~\cite[Thm.~1]{tarokh1998criteria}, the dominant
minimum-distance term gives
\[
\mathrm{CER} \approx C\Big(1+\frac{\Delta_{\min}}{4N_0}\Big)^{-2}
= C\Big(1+\frac{\Delta_{\min}}{4}\frac{\mathrm{SNR}}{E_s}\Big)^{-2},
\]
where $\mathrm{SNR}=E_s/N_0$ and $E_s$ is the average symbol energy. For a fixed
target $\varepsilon\in(0,C)$ we obtain \( \mathrm{SNR}(\varepsilon)
=\frac{4E_s}{\Delta_{\min}}\big(\sqrt{C/\varepsilon}-1\big),
\)
so if two constellations have the same $\Delta_{\min}$ and $C$, then
\(
\mathrm{SNR}_{\Z[i]}(\varepsilon)/\mathrm{SNR}_{\Z[w]}(\varepsilon)= E_{s,\Z[i]}/E_{s,\Z[w]}.
\)

It remains to compare $E_s$ for the square and hexagonal Voronoi cells. Let
$\rho$ be the apothem of the coarse cell $V=\mathrm{Vor}(\Lambda_c)$, with
$\Lambda_c=p\Lambda_f$ and $\Lambda_f=\Z^2$ or $A_2$. For large $p$, $E_s$
converges to the second moment per unit area of $V$~\cite{Con2013,zamir2014lattice}. For $\Lambda_f=\Z^2$, $V$ is a centred square of side $s=2\rho$:
\[
E_{s,\Z[i]}
=\frac{1}{s^2}\int_{-s/2}^{s/2}\!\!\int_{-s/2}^{s/2}(x^2+y^2)\,dx\,dy
=\frac{s^2}{6}=\frac{2}{3}\rho^2.
\]

For $\Lambda_f=A_2\cong\Z[w]$, $V$ is a regular hexagon with apothem $\rho$. With side $s=2\rho/\sqrt{3}$ its area is
\(
A=\frac{3\sqrt{3}}{2}s^2
=\frac{3\sqrt{3}}{2}\cdot\frac{4\rho^2}{3}
=2\sqrt{3}\,\rho^2.
\)
Splitting $V$ into six sectors and using polar coordinates, in one sector
\[
\int_{-\pi/6}^{\pi/6}\!\int_{0}^{\rho/\cos\theta} r^3\,dr\,d\theta
=\frac{1}{4}\rho^4\int_{-\pi/6}^{\pi/6}\sec^4\theta\,d\theta
=\frac{1}{4}\rho^4\frac{20}{9\sqrt{3}}.
\]
Multiplying by $6$ sectors and dividing by $A$ gives
\[
E_{s,\Z[w]}
=\frac{1}{A}\iint_V \|x\|^2\,dx
=\frac{\frac{10}{3\sqrt{3}}\rho^4}{2\sqrt{3}\rho^2}
=\frac{5}{9}\rho^2.
\]

Hence
\(
\frac{E_{s,\Z[i]}}{E_{s,\Z[w]}}
=\frac{(2/3)\rho^2}{(5/9)\rho^2}
=\frac{6}{5},
\) so
\[
\Delta\mathrm{SNR}
=10\log_{10}\frac{\mathrm{SNR}_{\Z[i]}(\varepsilon)}{\mathrm{SNR}_{\Z[w]}(\varepsilon)}
=10\log_{10}\frac{6}{5}\approx 0.79~\mathrm{dB}.
\]
\end{proof}

%\vspace{-.3cm}
%These values are consistent with the tabulated normalized second moments (NSM) of $\Z^2$ and $A_2$~\cite{Con2013,zamir2014lattice}:
The $\approx 0.79$~dB advantage of $\Z[w]$ over $\Z[i]$ is a purely geometric effect due to the smaller second moment of the hexagonal Voronoi cell. %(the classical “hexagon vs.\ square” shaping gain). %This theoretical prediction is confirmed by numerical evaluation. Table~\ref{tab:energy_gain} lists representative average symbol energies and SNR gains for finite constellations of size $K=p^2$. %Already at $p=13$ and $p=37$ the observed gains ($0.75$ and $0.787$~dB) are very close to the asymptotic $0.79$~dB limit.

% \begin{table}[htbp]
%     \centering
%     \caption{Average symbol energies and shaping gain for Alamouti--Eisenstein (hexagonal) vs.~Alamouti Classic (square) constellations.}
%     \begin{tabular}{|c|c|c|c|c|}
%         \hline
%         $p$ & $K=p^2$ & $E_{s,\Z[w]}$ & $E_{s,\Z[i]}$ & Gain (dB) \\
%         \hline
%         13 & 169  & 23.54    & 28.00   & 0.754  \\
%         37 & 1369 & 190.22   & 228.00  & 0.787 \\
%         61 & 3721 & 516.885  & 620.00  & 0.790 \\
%         73 & 5329 & 740.219  & 888.00  & 0.791 \\ 
%         97 & 9409 & 1306.887 & 1568.00 & 0.791 \\
%         \hline
%     \end{tabular}
%     \label{tab:energy_gain}
% \end{table}
\vspace*{-.3cm}
\subsection{Mutual Information}

We now compare the constellation-constrained mutual information (MI) of the
Alamouti--Eisenstein and classical Alamouti schemes. After Alamouti combining
over a coherent $2\times1$ Rayleigh channel, the effective scalar channel is
$\widetilde Y = X + Z$ with conditional noise variance $N_0/H$, where
$H=|h_1|^2+|h_2|^2\sim\Gamma(2,1)$~\cite{tse2005fundamentals}. For equiprobable inputs $X\in\mathcal C$, the MI per complex symbol is
{\small \begin{align}\label{MI}
    I_{\mathcal C}\!\left(\tfrac{E_s}{N_0}\right)
    = \E_{H,X,\widetilde Y}\!\left[\log_2
    \frac{p_{\widetilde Y|X,H}(\widetilde Y|X,H)}{
    \tfrac{1}{|\mathcal C|}\sum_{x'\in\mathcal C}
    p_{\widetilde Y|X,H}(\widetilde Y|x',H)}\right].
\end{align}}

For large SNR, standard manipulations (grouping terms by inter-point distance)
give the approximation
\begin{equation}\label{MI2}
\log_2|\mathcal C| - I_{\mathcal C}(\text{SNR})
\approx \dfrac{1}{\ln2}\sum_{r\geq 1} N_{\mathcal C}(r)
\Big(1 + \tfrac{\text{SNR}\, d_r^2}{E_s}\Big)^{-2},
\end{equation}
where $d_r$ are the distinct distances and $N_{\mathcal C}(r)$ their average
multiplicities. Expanding at high SNR yields
\begin{equation}\label{MI3}
\log_2|\mathcal C| - I_{\mathcal C}(\text{SNR})
\approx \frac{1}{\ln2}\Big(\tfrac{E_s}{\text{SNR}}\Big)^2
S(\mathcal C),
\end{equation}
with fourth-power distance spectrum \( S(\mathcal C) = \sum_{r\geq 1} \frac{N_{\mathcal C}(r)}{d_r^4}.\)

\begin{proposition}[Asymptotic MI gap]\label{prop:MI_gap}
Consider two Alamouti constellations carved from $\Z[i]$ and $\Z[w]$ with Voronoi shaping, equiprobable signalling and the same cardinality. Let $I_{\mathcal C}(\mathrm{SNR})$ denote the constellation-constrained mutual information per complex symbol over the coherent $2\times 1$ Rayleigh channel. In the high-SNR and large-constellation regime, %the SNRs needed by the two schemes to achieve the same mutual information satisfy \(\mathrm{SNR}_{\Z[i]}/\mathrm{SNR}_{\Z[w]}\approx 0.7816\) and therefore
we have \(\Delta\mathrm{SNR} \approx\;0.257~\mathrm{dB}.\)
%\[ \Delta\mathrm{SNR} \;\approx\;10\log_{10}(1.2)\;+\;5\log_{10}(0.7816) \;\approx\;0.257~\mathrm{dB}. \]
\end{proposition}

\begin{proof}
% From~\eqref{MI2}, the mutual information deficit at high SNR is
% \[
% \log_2|\mathcal C| - I_{\mathcal C}(\mathrm{SNR}) \approx \frac{1}{\ln2}\sum_{r\geq 1} N_{\mathcal C}(r) \Big(1 + \tfrac{\mathrm{SNR}\, d_r^2}{E_s}\Big)^{-2}.
% \]
%where $d_r$ and $N_{\mathcal C}(r)$ denote the distinct distances and their average multiplicities.
From~\eqref{MI2}, for large SNR, we get
\( \Big(1 + \tfrac{\mathrm{SNR}\, d_r^2}{E_s}\Big)^{-2} \sim \Big(\tfrac{E_s}{\mathrm{SNR}}\Big)^2 \frac{1}{d_r^4}. \) Then, we obtain 
\begin{equation}\label{MI4}
\log_2|\mathcal C| - I_{\mathcal C}(\mathrm{SNR}) \approx \frac{1}{\ln2} \Big(\tfrac{E_s}{\mathrm{SNR}}\Big)^2 S(\mathcal C),
\end{equation}
with fourth-power distance spectrum $S(\mathcal C)=\sum_{r\geq 1} \frac{N_{\mathcal C}(r)}{d_r^4}$.

Fix a target deficit $\Delta I>0$ and require $\log_2|\mathcal C|-I_{\mathcal C}(\mathrm{SNR})=\Delta I$ for both constellations. Then~\eqref{MI4} gives
% \[
% \Big(\tfrac{E_{s,\Z[i]}}{\mathrm{SNR}_{\Z[i]}}\Big)^2 S(\mathcal C_{\Z[i]}) \;\approx\; \Big(\tfrac{E_{s,\Z[w]}}{\mathrm{SNR}_{\Z[w]}}\Big)^2 S(\mathcal C_{\Z[w]}), 
% \]
% and hence
\[
\frac{\mathrm{SNR}_{\Z[i]}}{\mathrm{SNR}_{\Z[w]}} \;\approx\; \frac{E_{s,\Z[i]}}{E_{s,\Z[w]}} \sqrt{\frac{S(\mathcal C_{\Z[i]})}{S(\mathcal C_{\Z[w]})}}.
\]
As the constellation size grows, $\mathcal C_{\Z[i]}$ and $\mathcal C_{\Z[w]}$ approach the lattices $\Z^2$ and $A_2$, and $S(\mathcal C)$ converges to the Epstein–zeta value~\cite[Ch.~2]{Con2013} at $s=2$: \(
S(\mathcal C_{\Z[i]}) \to \zeta_{\Z^2}(2) =\sum_{(m,n)\neq(0,0)}\frac{1}{(m^2+n^2)^2},
\) and \( S(\mathcal C_{\Z[w]}) \to \zeta_{A_2}(2) =\sum_{(m,n)\neq(0,0)}\frac{1}{(m^2+mn+n^2)^2}.
\)

Numerically, $\zeta_{\Z^2}(2)\approx 6.0264$ and $\zeta_{A_2}(2)\approx 7.7711$, so \(
\frac{\zeta_{\Z^2}(2)}{\zeta_{A_2}(2)} \approx 0.7816.
\)
From Proposition~\ref{prop:shaping_gain} we have $E_{s,\Z[i]}/E_{s,\Z[w]}\to 6/5=1.2$. Substituting into the SNR ratio gives
% \[
% \frac{\mathrm{SNR}_{\Z[i]}}{\mathrm{SNR}_{\Z[w]}} \;\approx\;1.2\sqrt{0.7816},
% \]
%so in dB,
\[
\Delta\mathrm{SNR} \approx 10\log_{10}(1.2) +5\log_{10}(0.7816) \approx 0.257~\mathrm{dB}, 
\] 
\end{proof}

\subsection{Finite-Blocklength Achievable Rate}

In practice, the codeword length is limited by latency, so the asymptotic
condition $R<I(X;\widetilde Y)$ is not enough to predict reliability. In the
short-block regime, performance is governed by both the mutual information and
the dispersion of the channel~\cite{polyanskiy2010finite}.

As in the previous subsection, after Alamouti combining over a coherent
$2\times 1$ Rayleigh channel, each complex symbol experiences the effective
scalar channel
\begin{equation}\label{eq:alamouti_scalar}
\widetilde Y = X + Z,\qquad
Z\,|\,H\sim\mathcal{N}_c\!\left(0,\frac{N_0}{H}\right),
\end{equation}
where $X$ is uniform over the shaping constellation $\mathcal X\subset\C$,
$H=|h_1|^2+|h_2|^2\sim\Gamma(2,1)$, and the receiver knows $H$.

For one complex symbol, the conditional information density is
\(
i(X;\widetilde Y\mid H)
= \log\frac{p_{\widetilde Y|X,H}(\widetilde Y|X,H)}{p_{\widetilde Y|H}(\widetilde Y|H)}.
\)
% and the constellation-constrained mutual information is
% \[
% I(X;\widetilde Y) = \E\!\big[i(X;\widetilde Y\mid H)\big],
% \]
% where the expectation is over $(H,X,\widetilde Y)$.

A codeword consists of $m$ Alamouti blocks, with two symbols per block, so the
length is $n=2m$ complex symbols. In block $t$, the two symbols
$X_{t,0},X_{t,1}$ are independent and uniform over $\mathcal X$, and after
combining we obtain
\[
\widetilde Y_{t,j} = X_{t,j} + \widetilde Z_{t,j},\qquad j\in\{0,1\},
\]
with $\widetilde Z_{t,j}\,|\,H_t\sim\mathcal{N}_c(0,N_0/H_t)$ and
$\{H_t\}_{t=1}^m$ i.i.d.\ $\Gamma(2,1)$.

Define the information density of block $t$ as
\[
i_t := i(X_{t,0};\widetilde Y_{t,0}\mid H_t)
     + i(X_{t,1};\widetilde Y_{t,1}\mid H_t).
\]
Then $\{i_t\}_{t=1}^m$ are i.i.d., and the total information density over a
codeword is $\sum_{t=1}^m i_t$. The mean and variance per block are
\begin{equation}\label{eq:Iblk}
I_{\mathrm{blk}} := \E[i_t] = 2I(X;\widetilde Y),
\end{equation}
{\small \begin{align}\label{eq:Vblk_def}
V_{\mathrm{blk}} = 2\,\E[\mathrm{Var}(i(X;\widetilde Y\mid H)\mid H)]
  + 4\,\mathrm{Var}(\E[i(X;\widetilde Y\mid H)\mid H]).
\end{align}}
The effective dispersion per complex symbol is
\begin{equation}\label{eq:V_def}
V(X;\widetilde Y) := \frac{V_{\mathrm{blk}}}{2}.
\end{equation}

\begin{definition}[Finite-blocklength fundamental limit]
For blocklength $n$ symbols and target average error probability $p_e$, let
$M^*(n,p_e)$ be the largest number of codewords in an $(n,M,p_e)$ code over the
channel~\eqref{eq:alamouti_scalar}. The maximal achievable rate per complex
symbol is
\(
R^*(n,p_e) := \frac{1}{n}\log M^*(n,p_e).
\)
\end{definition}

\begin{theorem}[Normal-approximation achievable rate]\label{thm:normal_approx}
For the block-memoryless Alamouti equivalent channel above, with equiprobable
inputs $X\in\mathcal C$, i.i.d.\ Rayleigh fading $H\sim\Gamma(2,1)$ across
blocks, target error probability $p_e\in(0,1)$ and coherent ML decoding, the
maximal achievable rate satisfies
\[
R^*(n,p_e) \;\gtrsim\;
I(X;\widetilde Y) -\sqrt{\frac{V(X;\widetilde Y)}{n}}\,Q^{-1}(p_e)
+ O\!\left(\frac{\log n}{n}\right),
\]
where $I(X;\widetilde Y)$ and $V(X;\widetilde Y)$ are given by
\eqref{eq:Iblk}–\eqref{eq:V_def}.
\end{theorem}

\begin{proof}[Proof sketch]
The result follows from the dependence-testing bound and the Berry–Esseen
theorem applied to the i.i.d.\ sum $\sum_{t=1}^m i_t$, exactly as in
\cite[Th.~48]{polyanskiy2010finite}, after noting that each block contributes
two symbols and using the relations in \eqref{eq:Iblk}–\eqref{eq:V_def}. 
\end{proof}

\vspace*{-.3cm}

The same normal approximation applies to the classical Alamouti scheme; only $I(X;\widetilde Y)$ and $V(X;\widetilde Y)$ change with the underlying constellation. At a fixed SNR, the Eisenstein design has larger mutual information and smaller dispersion than the classical Alamouti code. Hence, for any rate $R_0$ strictly below both mutual informations, the block error probability of the classical scheme is larger for moderate $n$, and the gap typically increases with $n$.

Figure~\ref{fig:finite_blocklength} shows the normal-approximation block error
probability $\varepsilon$ as a function of the rate $R$ at $E_s/N_0=22$~dB for
$n\in\{128,256,512,1024\}$ and $M=169$ points. Solid curves correspond to the
Alamouti–Eisenstein constellation, dashed curves to the classical Alamouti
constellation. For every $n$, the Alamouti–Eisenstein curves lie below and to
the right, indicating a better rate–reliability trade-off. For example, at
$R_0=6.758$~bits/symbol and $n=256$, the proposed design gives
$\varepsilon \approx 0.0126$ versus $\varepsilon \approx 0.0325$ for the
classical scheme; at $n=1024$, the gap becomes
$3.77\times 10^{-6}$ versus $1.12\times 10^{-4}$. This behaviour is consistent
with Theorem~\ref{thm:normal_approx}: the Alamouti–Eisenstein constellation has
slightly larger mutual information and smaller dispersion, so its advantage
becomes more pronounced as $n$ increases.

\vspace*{-0.3cm}
\begin{figure}[ht]
    \centering
    \includegraphics[width=.9\linewidth]{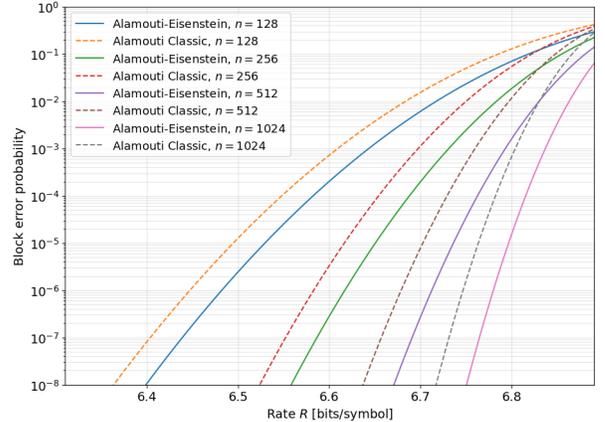}
    \caption{Normal-approximation $\varepsilon$--$R$ curves at $E_s/N_0=22$~dB: Alamouti–Eisenstein (solid) versus classical
    Alamouti (dashed), with $n\in\{128,256,512,1024\}$ symbols.}
    \label{fig:finite_blocklength}
\end{figure}

\vspace*{-0.3cm}

\section{Conclusions and Perspectives}\label{sec:conclusions}

We studied the classical Alamouti orthogonal design with Eisenstein (HEX) symbols from a quaternionic viewpoint. Starting from the definite quaternion algebra $(-1,-3)_{\Q}$, we exhibited an explicit maximal order $\Gamma \cong \Z[w]\oplus i\,\Z[w]$ whose canonical embedding into $\C^{2\times 2}$ yields an Alamouti--Eisenstein $2\times 2$ code with entries in $\Z[w]$. This gives a structural maximal-order interpretation of the Eisenstein realisation of the Alamouti code. The design preserves orthogonality, full diversity, and the non-vanishing determinant property. The associated $4$-dimensional embedded lattice has $A_2\oplus A_2$ geometry, and thus inherits the shaping gain of hexagonal packing.

We compared this construction with the classical Alamouti scheme over $\Z[i]$. With uniform signalling, we obtained an asymptotic shaping gain of about $0.79$\,dB and a small but positive gain in constellation-constrained mutual information. A finite-blocklength normal-approximation analysis for the equivalent Alamouti channel further indicated slightly higher mutual information and smaller dispersion, leading to improved short-packet reliability at the same SNR and rate.

Future work includes layered space--time index coding, extensions to other algebraic space--time designs, and a more detailed study of decoding complexity and practical trade-offs between the Eisenstein and Gaussian versions of the Alamouti scheme, as well as short-packet reliability under realistic overhead constraints~\cite{ostman2019short-packets,durisi2016mURLL}.

\bibliographystyle{IEEEtran}
\bibliography{references}

\end{document}